\documentclass[letterpaper, 10 pt, conference]{ieeeconf}  % Comment this line out if you need a4paper
\IEEEoverridecommandlockouts
\usepackage{cite}
\usepackage{amsmath,amssymb,amsfonts}
\usepackage{amsmath}
\usepackage{textcomp}
\usepackage{xcolor}
\usepackage{enumerate}   
\usepackage{amsmath}
\usepackage[utf8]{inputenc} % allow utf-8 input
\usepackage[T1]{fontenc}    % use 8-bit T1 fonts
\usepackage{hyperref}       % hyperlinks
\usepackage{url}            % simple URL typesetting
\usepackage{booktabs}       % professional-quality tables
\usepackage{amsfonts}       % blackboard math symbols
\usepackage{nicefrac}       % compact symbols for 1/2, etc.
\usepackage{microtype}      % microtypography
\usepackage{amsfonts}
\usepackage{amssymb}
\usepackage{graphicx}
\usepackage{fancyhdr,color}
\usepackage{mathtools}
\usepackage{bbm}
\usepackage{bm}
\usepackage{mathabx}
\usepackage{balance}

\newcommand*{\medcap}{\mathbin{\scalebox{0.8}{\ensuremath{\bigcap}}}}%
\newcommand*{\medcup}{\mathbin{\scalebox{0.8}{\ensuremath{\bigcup}}}}%
\newtheorem{theorem}{Theorem}
\newtheorem{remark}{Remark}

\newtheorem{proposition}{Proposition}

\usepackage[ruled, lined, linesnumbered, commentsnumbered, longend]{algorithm2e}
\def\BibTeX{{\rm B\kern-.05em{\sc i\kern-.025em b}\kern-.08em
    T\kern-.1667em\lower.7ex\hbox{E}\kern-.125emX}}
\begin{document}

\title{\LARGE \bf Eco-driving under localization uncertainty for connected vehicles on Urban roads: Data-driven approach and Experiment verification}

\author{Eunhyek Joa \qquad Eric Yongkeun Choi \qquad Francesco Borrelli
% \IEEEauthorblockA{$^{1}$Mechanical Engineering, University of California, Berkeley, USA}
\thanks{The authors are with the Model Predictive Control Lab,
Department of Mechanical Engineering, University of California at Berkeley (e-mail: e.joa@berkeley.edu; yk90@berkeley.edu; fborrelli@berkeley.edu). }}

\maketitle

\begin{abstract}
This paper addresses the eco-driving problem for connected vehicles on urban roads, considering localization uncertainty.
Eco-driving is defined as longitudinal speed planning and control on roads with the presence of a sequence of traffic lights.
We solve the problem by using a data-driven model predictive control (MPC) strategy.
This approach involves learning a cost-to-go function and constraints from state-input data.
The cost-to-go function represents the remaining energy-to-spend from the given state, and the constraints ensure that the controlled vehicle passes the upcoming traffic light timely while obeying traffic laws. 
The resulting convex optimization problem has a short horizon and is amenable for real-time implementations.
We demonstrate the effectiveness of our approach through real-world vehicle experiments. 
Our method demonstrates $12\%$ improvement in energy efficiency compared to the traditional approaches, which plan longitudinal speed by solving a long-horizon optimal control problem and track the planned speed using another controller, as evidenced by vehicle experiments.
\end{abstract}

% \begin{IEEEkeywords}
% Connected Autonomous Vehicles, Optimization-based Trajectory Planning, Vehicle-in-the-Loop, Vehicle Tests.  
% \end{IEEEkeywords}

\section{Introduction}
In the emerging era of smart cities, connected and automated vehicle (CAVs) technology offers significant advantages for daily driving and urban transportation \cite{ELLIOTT2019Recent}.
Recognizing its potential, extensive research efforts have been dedicated to advancing CAV technology, which promises to enhance road utilization, vehicle energy efficiency, and traffic safety \cite{guanetti2018CAVs, yao2021Fuel, karbasi2022Investigating}. 
These improvements are largely attributable to reducing human errors and integrating comprehensive traffic information provided by connectivity technologies.

In this paper, we focus on the aspect of vehicle energy efficiency among the benefits of CAV technology, addressing what is commonly referred to as eco-driving problems \cite{sciarretta2015optimal}.
Optimization-based algorithms have become the cornerstone for addressing eco-driving problems, and the primary objective of these algorithms is to minimize stops, along with unnecessary acceleration and deceleration, thereby significantly enhancing energy efficiency \cite{sciarretta2020CAVs}. 
This goal is pursued by developing optimal control problems that incorporate both vehicle and powertrain dynamics, as well as the surrounding traffic agents such as surrounding vehicles and traffic lights.

Recent research has explored a two-level control architecture to address eco-driving challenges, effectively reducing model complexity by splitting the optimization problem into two manageable sub-problems. 
This approach has been employed in various optimization control methods, including Dynamic Programming (DP) \cite{bae2019real, sun2020optimal}, Pontryagin's Minimum Principle (PMP) \cite{ard2023VILCAV, han2023energy}, and Model Predictive Control (MPC) \cite{chada2020ecological}, which all aid in lessening the computational load for easier implementation. 
However, this bifurcation into sub-problems often results in sub-optimal outcomes and issues such as delay and latency. 
To address these drawbacks, there have been initiatives to resolve the problem within a singular control layer using deep reinforcement learning (DRL) \cite{Vindula2022RL, bai2022RL, li2022RL}. 
Despite its promise, the sensitivity of DRL to parameter tuning and safety guarantees poses significant difficulties when applied in practical, real-world applications.

In practice, most eco-driving experiments leverage the two-level control architecture for the reasons mentioned above. 
However, as we approach the real-world application of connected and automated vehicles (CAVs), two primary issues arise. 
Firstly, the two-layer control architecture may introduce additional challenges beyond sub-optimality. 
Delays and latency discrepancies between the two layers can result in spatial and temporal misalignments, leading to conflicting control layer decisions. 
Furthermore, since each layer may prioritize different objectives (such as traffic light distance, speed limits, etc.), this can lead to the different tuning of controllers and thus varying vehicle performance. 
Secondly, the issue of vehicle localization uncertainty poses a challenge. 
Most existing studies presume the availability of precise position information, an assumption that may not hold in real-world scenarios. 
Upcoming CAV applications might have to rely on localization modules with lower accuracy, introducing non-negligible uncertainties, especially in scenarios like eco-approach at signalized intersections. 
Additionally, the computing power of most commercial vehicles is often insufficient for processing complex algorithms in real-time.

To address these challenges, we have developed a data-driven approach tailored to CAV eco-driving scenarios. This approach is designed to be computationally feasible and robust against localization uncertainties, offering a practical solution for the next imminent generation of CAV eco-driving technologies.
Our contributions are summarized as:
\begin{itemize}
    \item  We propose a novel real-time, data-driven MPC to approximately solve an eco-driving problem under localization uncertainty for connected vehicles.
    \item The proposed MPC can ensure the controlled vehicle's timely crossing of traffic lights within a user-defined duration in closed-loop, distinguishing it from previous research \cite{bae2019real, ard2023VILCAV, sun2020optimal}, which could only ensure such crossings in open-loop or during the planning phase.
    \item The proposed MPC is a convex optimization problem that can be solved using an off-the-shelf solver.
    \item We experimentally demonstrate the energy saving of the proposed algorithm through vehicle tests where the actual test vehicle is controlled to finish the given route under localization uncertainty while interacting with virtual, deterministic traffic lights.
\end{itemize}

\textit{Notation:} 
Throughout the paper, we use the following notation.
$\mathbf{0}^{n \times m}$ represents an $n$-by-$m$ zero matrix.
The positive semi-definite matrix \(P\) is denoted as \(P\succeq0\).
The Minkowski sum of two sets is denoted as \(\mathcal{X} \oplus \mathcal{Y} = \{x+y: x \in \mathcal{X}, y \in \mathcal{Y}\}\).
The Pontryagin difference between two sets is defined as \(\mathcal{X} \ominus \mathcal{Y} = \{x\in\mathcal{X}: x+y \in \mathcal{X}, \forall y \in \mathcal{Y}\}\).
The m-th column vector of a matrix \(H\) is denoted as \([H]_m\). The m-th component of a vector \(h\) is \([h]_m\). 
% \(\mathbb{P}(\mathcal{A})\) is the probability of the event \(\mathcal{A}\), and \(\mathbb{E}[\cdot]\) is the expectation of its argument. 
The notation \(x_{l:m}\) means the sequence of the variable \(x\) from time step \(l\) to time step \(m\).
% $\mathrm{dim}(v)$ denotes the dimension of the vector $v$. $\mathrm{int}(\mathcal{S})$ denotes the interior of the set $\mathcal{S}$. $\mathrm{conv}\{\mathcal{S}\}$ denotes the convex hull of the set $\mathcal{S}$. 

\section{Problem Setup}
In this section, we formulate the problem of eco-driving for connected vehicles on urban roads under localization uncertainty. 
We assume that a route from a starting point A to a goal point B is given.
Moreover, we assume that traffic light cycles are deterministic, and the road grade is small enough to neglect the gravitational potential energy in energy consumption.
\subsection{Route segmentation} \label{sec: urban driving parameterization}
Our route segmentation strategy is illustrated in Fig. \ref{fig:road_segmentation}.
\begin{figure}[ht]
%\vspace{-1.0em}
\begin{center}
\includegraphics[width=0.85\linewidth,keepaspectratio]{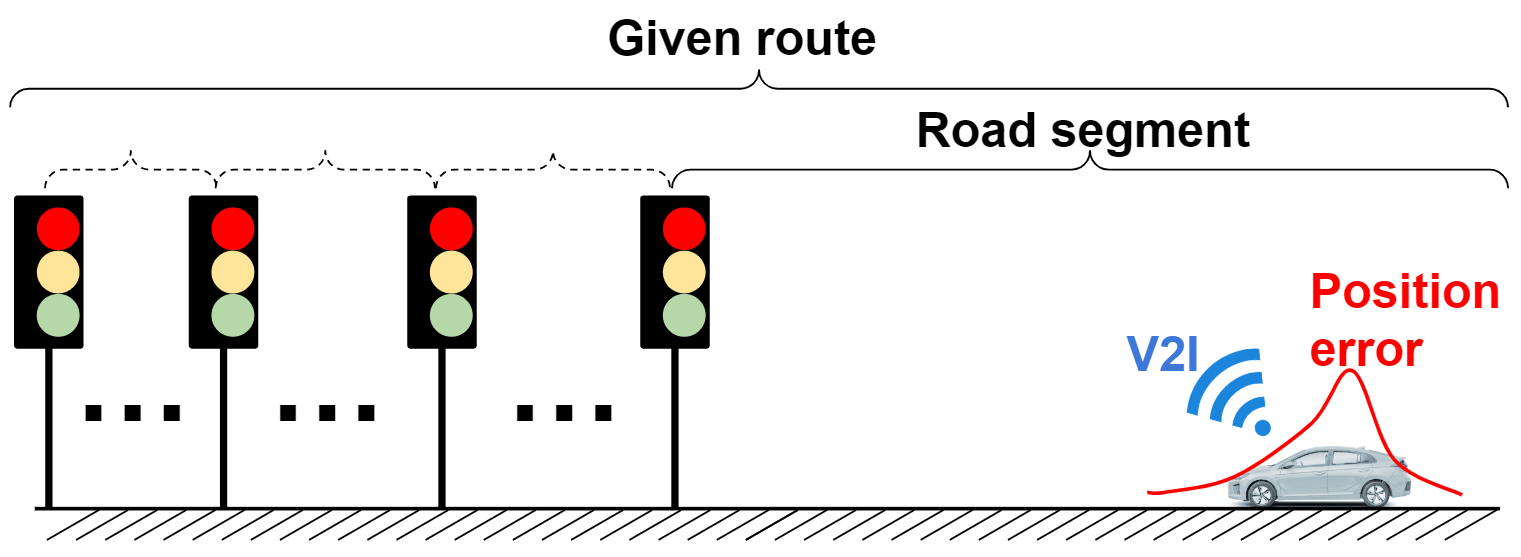}
\vspace{-1.0em}
\caption{The given route is a series of multiple, parameterized road segments.}
\label{fig:road_segmentation}
\end{center}
\vspace{-1.0em}
\end{figure}\\
We define a road segment as the combined section comprising a traffic light and the connected road.
The parameters defining this segment are the current traffic light signal, a traffic light cycle, the remaining time of the current traffic light signal, and the distance to the traffic light.
When considering the route from point A to point B in urban road driving, we consider it as a series of multiple road segments, each distinguished by varying parameters. 

\subsection{Vehicle Model, Measurement Model, and Observer}
We model vehicle longitudinal dynamics as follows:
%\vspace{-0.12cm}
\begin{equation} \label{eq: longitudinal dynamics}
\begin{aligned}
    & \mathbf{x}(t) = \begin{bmatrix} s(t) & v_x(t) \end{bmatrix}^\top, ~ u(t) = a_x(t), \\
    & \Dot{\mathbf{x}}(t) = \begin{bmatrix} 0 & 1 \\ 0 & 0 \end{bmatrix}\mathbf{x}(t) + \begin{bmatrix} 0 \\ 1 \end{bmatrix}u(t),
\end{aligned}
%\vspace{-0.12cm}
\end{equation}
where the state $s$ represents a relative longitudinal position along the centerline of the given route, and $v_x$ and $a_x$ are vehicle longitudinal speed and acceleration, respectively.
Throughout the paper, we will refer to $s$ as the position for brevity. 
Forces due to road grade, air drag, and rolling resistance are not included in \eqref{eq: longitudinal dynamics} because $a_x$ is net longitudinal acceleration. 
We regard controlling the vehicle under those forces as the task of the actuator-level controller.

We discretize the dynamics \eqref{eq: longitudinal dynamics} as
%\vspace{-0.12cm}
\begin{equation} \label{eq:system}
\begin{aligned}
    & \mathbf{x}_{k}=\begin{bmatrix} s_k & v_{x,k} \end{bmatrix}^\top, ~ u_{k} = a_{x,k}, \\
    & \mathbf{x}_{k+1} = \mathbf{A} \mathbf{x}_{k} + \mathbf{B} u_{k},
\end{aligned}
%\vspace{-0.12cm}
\end{equation}
where $\mathbf{x}_{k}$ denotes the state, and $u_{k}$ denotes the input at time step $k$. $s_k$ represents the position, while $v_x$ and $a_x$ correspond to the vehicle's longitudinal speed and acceleration at time step $k$. The discretization sampling time $T_s$ is 1 sec.

At time step $k$, we measure the states from sensors as:
%\vspace{-0.12cm}
\begin{equation} \label{eq: sensor measurements}
\begin{aligned}
    & \mathbf{y}_k = \mathbf{C} \mathbf{x}_{k} + \mathbf{D} w_k = \begin{bmatrix} 1 & 0 \\ 0 & 1\end{bmatrix}\mathbf{x}_{k} + \begin{bmatrix} 1 \\ 0 \end{bmatrix}w_k,
\end{aligned}
%\vspace{-0.12cm}
\end{equation}
where $\mathbf{y}_k$ denotes sensor measurements at time step $k$. The position $s$ is measured by localization modules and $w$ represents its localization uncertainty. The speed $v_x$ is measured by vehicle wheel encoders, 
We assume vehicle speed measurement noises are small enough to neglect them.
We assume that the localization uncertainty $w$ is a random variable with distribution $p(w)$ and bounded convex polytope $\mathcal{W}$ as follows:
%\vspace{-0.12cm}
\begin{equation} \label{eq: gps error}
\begin{aligned}
    & w_k \sim p(w), ~ w_k \in \mathcal{W}.
\end{aligned}
%\vspace{-0.12cm}
\end{equation}
In practice, the localization uncertainty may be designed as a Gaussian distribution, which is not bounded. In this case, a high confidence interval can be used to approximate $\mathcal{W}$.

We design a discrete state observer \cite{ogata2010modern} with a predictor to estimate the state $\mathbf{x}_{k}$ as follows:
%\vspace{-0.12cm}
\begin{equation} \label{eq: observer equation}
\begin{aligned}
    & \hat{\mathbf{x}}_{0} = \mathbf{y}_0, \\
    & \hat{\mathbf{x}}_{k+1} = \mathbf{A} \hat{\mathbf{x}}_{k} + \mathbf{B} u_{k} + \begin{bmatrix} L & 0 \\ 0 & 1 \end{bmatrix}\biggr(\mathbf{y}_{k+1} - \mathbf{C} ( \mathbf{A} \hat{\mathbf{x}}_{k} + \mathbf{B} u_{k})\biggr), \\
    & ~~~~~~ =  \mathbf{A} \hat{\mathbf{x}}_{k} + \mathbf{B} u_{k} + \mathbf{F}n_k,
\end{aligned}
%\vspace{-0.12cm}
\end{equation}
where $L$ is an observer gain, $\mathbf{F} = \begin{bmatrix} 1 & 0 \end{bmatrix}^\top$, and $n_k = L(s_{k} -\hat{s}_{k}) + Lw_{k+1}$, which is a lumped noise. Note that as the speed measurement is accurate, the observer \eqref{eq: observer equation} is designed to set the current speed estimate to the current speed measurement, i.e., $\hat{v}_{x,k} = \begin{bmatrix} 0 & 1\end{bmatrix}y_k$. On the other hand, as the position measurement is not accurate, the observer \eqref{eq: observer equation} is designed to suppress the localization uncertainty.
% Note that as the speed measurement is accurate, the observer is designed to suppress the localization uncertainty in the position, and to set the current speed estimate $\hat{v}_{x,k} = \begin{bmatrix} 0 & 1\end{bmatrix}y_k$.

\begin{proposition} \label{prop: delta s and lumped noise}
    Let $\Delta s_k = s_{k} -\hat{s}_{k}$. Then, $\Delta s_k \in \mathcal{W}$ and $n_k \in 2L\mathcal{W}$ for all realizations of noise that satisfies \eqref{eq: gps error}.
\end{proposition}
% %\vspace{-1.0em}
\begin{proof}
From \eqref{eq:system} and \eqref{eq: observer equation}, we have that:
%\vspace{-0.12cm}
\begin{equation} \label{eq: estimation error dyn}
\begin{aligned}
    & \hat{s}_0 = \begin{bmatrix}1 & 0 \end{bmatrix} y_0, ~\Delta s_0 =  -w_0, \\
    & \Delta s_{k+1} = (1-L)\Delta s_{k} - Lw_{k+1}.
\end{aligned}
%\vspace{-0.12cm}
\end{equation}
Proof by induction.
First, $\Delta s_0 = -w_0 \in \mathcal{W}$.
Now, let $\Delta s_{k} \in \mathcal{W}$. Since $\Delta s_{k+1} = (1-L)\Delta s_{k} - Lw_{k+1} \in (1-L) \mathcal{W} \oplus L \mathcal{W} = \mathcal{W}$, $\Delta s_{k+1} \in \mathcal{W}$. Thus, $\Delta s_{k} \in \mathcal{W}$ for all $k \geq 0$ by induction. 
Furthermore, $n_k = L\Delta s_k + Lw_{k+1} \in L\mathcal{W} \oplus L\mathcal{W} = 2L\mathcal{W}$.
\end{proof}
% by setting $L=0.5$, we can ensure that $n_k \in \mathcal{W}$ for all $\mathcal{W}$.
The lumped noise $n_k$ is a random variable with distribution $q(n)$ and bounded support $2L\mathcal{W}$ as follows:
%\vspace{-0.12cm}
\begin{equation} \label{eq: lumped noise}
\begin{aligned}
    & n_k \sim q(n), ~ n_k \in 2L\mathcal{W}.
\end{aligned}
%\vspace{-0.12cm}
\end{equation}
$q(n)$ can be explictly calculated by following \cite[eq. (14), (15)]{ejoaof} when $p(w)$ is Gaussian.
In practice, differential GPS, which offers centimeter-level accuracy in positioning, can be used to calculate $n_k$ and approximate its density function $q(n)$.

\subsection{Energy Consumption Stage Cost}
We define the energy $E$ as the sum of the energy stored in battery and/or fuel tank  $E_\text{stor}$ and the kinetic energy $E_\text{kin}$ as:
%\vspace{-0.12cm}
\begin{equation} \label{eq: energy definition}
\begin{aligned}
    & E(t) = E_\text{stor}(t) + E_\text{kin}(t).
\end{aligned}
%\vspace{-0.12cm}
\end{equation}
The energy consumption at time step $k$ is defined as the change in energy between time step $k$ and $k+1$ as follows:
%\vspace{-0.12cm}
\begin{equation} \label{eq: power definition}
\begin{aligned}
    & \Delta E_k = E(kT_s) - E((k+1)T_s).
\end{aligned}
%\vspace{-0.12cm}
\end{equation}
\begin{remark}
    As the energy in \eqref{eq: energy definition} represents the total energy of the vehicle, which is the closed-system, the value of the total energy always decreases over time. Thus, the energy consumption in \eqref{eq: power definition} is nonnegative for all time step $k \geq 0$.  
\end{remark}
% where $m$ is a vehicle mass, $g$ is gravitational acceleration, and $h$ is a height of the vehicle position. $m g \Dot{h}(t)\simeq 0$ as we assume that the road grade is small enough to neglect the gravitational potential energy in energy consumption.
% Note that by the law of conservation of energy, the energy consumption \eqref{eq: power definition} is nonnegative.

We solve the following regression problem to obtain the parameterized energy consumption cost $\ell(\mathbf{x}_k, u_k)$ given data:
%\vspace{-0.12cm}
\begin{equation} \label{eq: energy cost regression}
\begin{aligned}
   & \min_{\mathbf{P} \in \mathbb{R}^{3\times3}} \sum_{k=0}^{T_\text{data}} \lVert \Delta E_k - \ell(\mathbf{x}_k, u_k) \rVert_2^2 \\
   & ~~~\text{s.t.,} ~~\ell(\mathbf{x}_k, u_k) = \begin{bmatrix} v_{x,k} & u_k & 1 \end{bmatrix} \mathbf{P} \begin{bmatrix} v_{x,k} \\ u_k \\ 1 \end{bmatrix}, ~\mathbf{P} \succeq 0,
\end{aligned}
%\vspace{-0.12cm}
\end{equation} 
where $T_\text{data}$ is the end time of the data, and $\mathbf{P}$ is a positive semi-definite matrix.
The parametrized cost is designed to be nonnegative as the energy consumption \eqref{eq: power definition} is nonnegative.
It is noteworthy that our parameterized cost does not depend on the position whose actual value cannot be obtained.
Thus, the following holds:
%\vspace{-0.12cm}
\begin{equation} \label{eq: CE stage cost}
\begin{aligned}
    & \ell(\mathbf{x}_k, u_k) = \ell(\hat{\mathbf{x}}_k, u_k).
\end{aligned}
%\vspace{-0.12cm}
\end{equation}

We collected the energy consumption data for our electric test vehicle using equipped sensors and solved the regression problem \eqref{eq: energy cost regression}. The energy consumption measurement data used for the regression was collected from 3.9 km of city driving scenarios at a testing track. The comparison between the data and the regressed model is given in Fig. \ref{fig:energy_cost_regression}. The regression results show that the energy consumption calculated from the regressed model is non-negative and, the error of the total energy consumption is below $1 \%$. The regression model validation results are presented in section \ref{sec: experiment}.
\begin{figure}[ht]
\begin{center}
\includegraphics[width=0.8\linewidth,keepaspectratio]{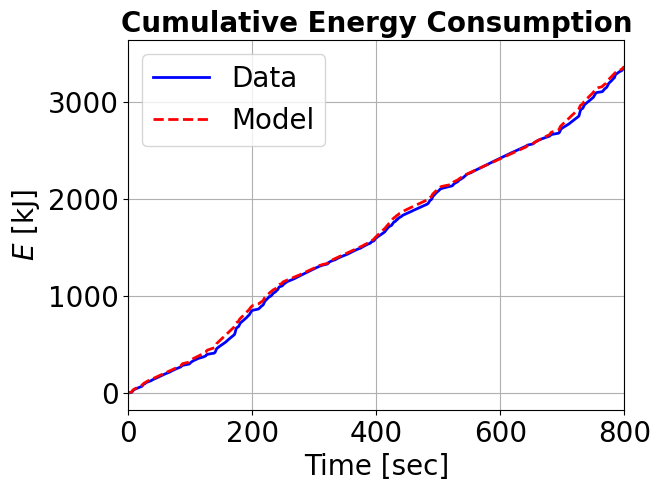}
\vspace{-1.0em}
\caption{Comparison: Vehicle energy measurement data and Simulated energy consumption}
\label{fig:energy_cost_regression}
\end{center}
\vspace{-1.0em}
\end{figure}

\subsection{Constraints}
System (\ref{eq:system}) is subject to the following constraints: % at all time steps:
%\vspace{-0.12cm}
\begin{equation} \label{eq:constraints}
\begin{aligned}
    & \mathbf{x}_{k} \in \mathcal{X} = \{\mathbf{x} ~|~ 0 \leq \begin{bmatrix} 0 & 1 \end{bmatrix} \mathbf{x} \leq v_{x, \max} \}, \\
    & u_k \in \mathcal{U} = \{u ~|~ a_{x, \min} \leq u \leq a_{x, \max} \}, \\
    & \forall k \geq 0, ~\forall w_k \in \mathcal{W},
\end{aligned} 
%\vspace{-0.12cm}
\end{equation}
where $v_{x, \max}$ is the maximum allowable speed, and $a_{x, \min, \max}$ is the minimum/maximum acceleration.
\(\mathcal{X}\) is a convex polyhedron, and \(\mathcal{U}\) is a polytope.
As the state constraints are imposed on the vehicle speed, $\mathbf{x}_{k} \in \mathcal{X}$ is identical to $\hat{\mathbf{x}}_{k} \in \mathcal{X}$.

\subsection{Eco-Driving Problem for Connected Vehicles under Localization Uncertainty}
The eco-driving problem for connected vehicles under localization uncertainty is a stochastic optimization problem due to stochastic localization uncertainty \eqref{eq: gps error}.
Specifically, it can be formulated as follows:
\begin{equation} \label{eq:ftocp}
\begin{aligned}
    & J^{\star}(\mathbf{x}_S) = \\
    & \min_{\Pi(\cdot)} ~\,\mathbb{E}_{w_{0:T_{f}}}\Biggr[\sum_{k=0}^{T_{f}} \ell(\hat{\mathbf{x}}_k, ~ \pi_k(\hat{\mathbf{x}}_k)) \Biggr] \\
    % & \qquad \qquad ~~ \textnormal{s.t.,} ~~ \hat{\mathbf{x}}_{k+1} =\mathbf{A} \hat{\mathbf{x}}_{k} + \mathbf{B} u_{k} + \begin{bmatrix} 1 \\ 0 \end{bmatrix}n_k, \\
    & ~ \textnormal{s.t.,} ~~ \mathbf{x}_{k+1}  = \mathbf{A} \mathbf{x}_k  + \mathbf{B} \pi_k (\hat{\mathbf{x}}_k), \\
    & \qquad ~ \mathbf{y}_{k}  = \mathbf{C} \mathbf{x}_k  + \mathbf{D}w_k , ~ w_k  \sim p(w), \\
    & \qquad ~ \hat{\mathbf{x}}_{k+1} =\mathbf{A} \hat{\mathbf{x}}_{k} + \mathbf{B} u_{k} \\
    & \qquad \qquad ~~~~ + \begin{bmatrix} L & 0 \\ 0 & 1 \end{bmatrix}\biggr(\mathbf{y}_{k+1} - \mathbf{C} ( \mathbf{A} \hat{\mathbf{x}}_{k} + \mathbf{B} u_{k})\biggr), \\
    & \qquad ~ \mathbf{x}_0 = \mathbf{x}_S, ~ \hat{\mathbf{x}}_{0} = \mathbf{y}_{0}, \\
    & \qquad ~ \mathbf{x}_{k} \in \mathcal{X}, ~ \pi_k(\hat{\mathbf{x}}_{k}) \in \mathcal{U},  ~\forall w_k \in \mathcal{W},  \\
    & \qquad ~ \mathbf{x}_{T_{f}} \in \mathcal{X}_f,\\
    & \qquad ~ k=0,...,T_{f}-1, 
\end{aligned}
\end{equation}
where $T_f\gg0$ is the task horizon and $\mathcal{X}_f$ represents the region beyond the goal point.
The cost function is an expected sum of the regressed energy consumption state cost $\ell(\cdot,\cdot)$ in \eqref{eq: energy cost regression} evaluated for the estimated state trajectory. From \eqref{eq: CE stage cost}, this cost is identical to the cost evaluated for the actual state trajectory.
% We point out that as the system \eqref{eq: observer equation} is uncertain, the optimal control problem (\ref{eq:ftocp}) consists of finding state feedback policies \(\Pi(\cdot) = \{\pi_0(\cdot),\pi_1(\cdot),...,\pi_{T_f-1}(\cdot)\} \).
% The cost function is an expected sum of the regressed energy consumption costs in \eqref{eq: energy cost regression}.
The system \eqref{eq:system} should satisfy the state and input constraints \eqref{eq:constraints} while minimizing the expected sum of costs.
% Moreover, the system \eqref{eq:system} should satisfy the specific times for passing upcoming traffic lights, which are calculated by the high-level module, while obeying the traffic light.
Solving this problem in real-time is challenging due to the long time horizon required to complete the given route and the presence of stochastic elements.

\subsection{Solution approach to the problem}
We take three approaches to solve this problem, namely: 
\begin{enumerate}[(i)]
    \item We consider the problem of crossing the current road segment within a specified time only at each step until the vehicle reaches goal point B.
    \item We solve a simpler constrained optimal control problem (OCP) with a short prediction horizon $N \ll T_f$ in a receding horizon fashion.
    \item We approximate the expected cost in \eqref{eq:ftocp} with a sample mean.
\end{enumerate}
(i) and (ii) are to alleviate computational costs due to the long time horizon required to complete the given route, while (iii) is to reformulate the expectation into a tractable term.

When we only consider the current road segment (as specified in (i)), the solution will not be necessarily relevant to the optimal solution of \eqref{eq:ftocp}.
To alleviate this issue, we utilize a high-level green wave search module such as in \cite{sun2020optimal, highlevelgreenwave1, highlevelgreenwave2}, which provides specific times for passing upcoming traffic lights without stopping.

When the horizon $N$ is restricted (as specified in (ii)), it might result in a prediction horizon that does not cover the entire current road segment. This limitation could lead to myopic behavior, as the optimization process focuses solely on the costs within the $N$ horizon stages. To prevent this myopic behavior arising from the limited prediction horizon, we design terminal constraints and a terminal cost function $V(\cdot)$. In this paper, we take a data-driven approach to obtain them.

\section{Method}
In this section, we introduce the data-driven approach to calculate the terminal constraints and a terminal cost function $V(\cdot)$ and present the proposed MPC.
\subsection{Data: State-Input pairs} \label{sec: data}
We generate the data set for the data-driven algorithm through two processes, initialization and augmentation.
\begin{remark}
    In this paper at Sec. \ref{sec: single road segment}, the data set is generated through simulation. However, it is worth mentioning that data generation is not exclusively tied to simulation; it can also be accomplished through closed-loop testing.
\end{remark}

To initialize the data set, we design a cruise control algorithm \cite[Sec. IV. C]{ejoaiv}, which satisfies the constraints \eqref{eq:constraints}. We collect state-input pairs while running the cruise control algorithm multiple times and construct the initial data set as follows:
\begin{equation} \label{eq: initial dataset}
    \mathcal{D} = \{(\mathbf{x}_d, ~u_d)\}_{d=0}^{N^0_\text{data}},
\end{equation}
where $N^0_\text{data}$ represents the number of the data points after the initialization process.

We recursively augment the data set starting from the initial data set \eqref{eq: initial dataset}.
We describe one iteration of the data augmentation process below.
Using the provided data, we construct two sets to ensure passing the traffic light within the specific time while obeying the traffic light as per Sec. \ref{sec: learning constraints}.
Later, an intersection of these two sets is utilized to define the terminal constraint. Moreover, we construct the terminal cost function as per Sec. \ref{sec: learning V function}. 
% We construct the terminal constraint by imposing the terminal state of the MPC.
We design the proposed MPC by imposing the terminal constraint on its terminal state and adding the terminal cost as per Sec. \ref{sec: Data-driven MPC}.
Subsequently, we run the proposed MPC.
During this procedure, the proposed MPC can be infeasible when a set that satisfies the terminal constraints is empty due to an insufficient amount of data.
In this case, we employ the cruise control algorithm \cite[Sec. IV. C]{ejoaiv} as a backup controller, which ensures the recursive feasibility of the proposed control algorithm.
We collect state-input pairs at each time step, and augment the data set as follows:
\begin{equation}
    \mathcal{D} \leftarrow \mathcal{D} \medcup \{(\mathbf{x}_d, ~u_d)\}_{d=0}^{N^{a}_\text{data}},
\end{equation}
where $N^{a}_\text{data}$ represents the number of augmenting data points.
This data augmentation process iterates until the closed-loop performance is settled, which will be shown in Sec. \ref{sec: single road segment}.

\subsection{Learning the terminal constraint} \label{sec: learning constraints}
In this section, we design two sets in a data-driven way, and the intersection of these two sets will define the terminal constraint of the proposed MPC in Sec. \ref{sec: Data-driven MPC}

We adopt the notion of \textit{Robust Controllable Set} to design these two sets.
A set is $N$-step \textit{Robust Controllable}, if all states belonging to the set can be robustly driven against additive noises, through a time-varying control law, to the target set in $N$-step \cite[Def.10.18]{borrelli2017predictive}.
We design two \textit{Robust Controllable Sets} against additive lumped noise $n$ \eqref{eq: lumped noise}:
\begin{enumerate} %[label=(S\arabic*)]
    \item $\mathcal{S}_{t_\text{red}}$: $t_\text{red}$-step Robust Controllable Set where the target set is the region before the traffic light, i.e., $\mathcal{T}_s = \{\mathbf{x} ~|~ \begin{bmatrix} 1 & 0 \end{bmatrix} \mathbf{x} \leq s_\text{tl}\}$.
    \item $\mathcal{P}_{t_\text{green}}$: $t_\text{green}$-step Robust Controllable Set where the target set is the region after the traffic light, i.e., $\mathcal{T}_p = \{\mathbf{x} ~|~ \begin{bmatrix} 1 & 0 \end{bmatrix} \mathbf{x} \geq s_\text{tl}\}$.
\end{enumerate}
$s_\text{tl}$ is the location of the upcoming traffic light.
$t_\text{red}$ and $t_\text{green}$ are defined in Fig. \ref{fig:definition of tred tgreen} where $k_\text{pass}$ denote the time to pass the traffic light specified by the high-level module. 

The intersection of these two sets, $\mathcal{S}_{t_\text{red}} \medcap \mathcal{P}_{t_\text{green}}$, is utilized to define the terminal constraint.
Let $\hat{\mathbf{x}}_{N|k}$ denotes a predicted estimated state at time step $k$+$N$ calculated from the current estimate $\hat{\mathbf{x}}_{k}$, i.e., a terminal state.
Then, the terminal constraint, $\hat{\mathbf{x}}_{N|k} \in \mathcal{S}_{t_\text{red}} \medcap \mathcal{P}_{t_\text{green}}$, enforces not to pass the traffic light before it turns green while guaranteeing that the vehicle crosses the traffic light within the time specified by the high-level module.
Fig. \ref{fig:planning_strategy} illustrates the terminal constraints. 
\begin{figure}[ht]
% %\vspace{-1.0em}
\begin{center}
\includegraphics[width=0.85\linewidth,keepaspectratio]{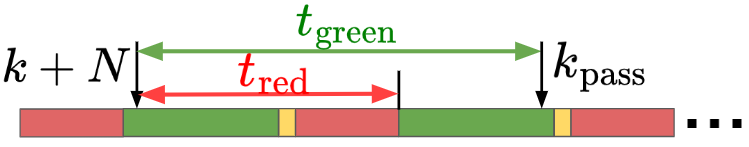}
%\vspace{-1.0em}
\caption{Defintion of $t_\text{red}$ and $t_\text{green}$}
\label{fig:definition of tred tgreen}
\end{center}
\end{figure}
% %\vspace{-3.0em}
\begin{figure}[ht]
\begin{center}
\includegraphics[width=0.85\linewidth,keepaspectratio]{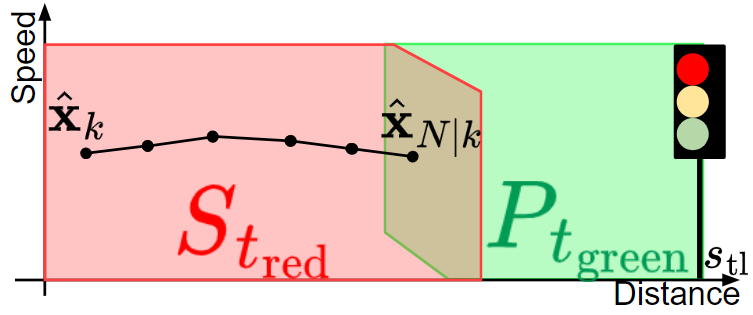}
%\vspace{-1.0em}
\caption{Illustration of the terminal constraints $\mathcal{S}_{t_\text{red}}$ and $\mathcal{P}_{t_\text{green}}$: the constraint $\hat{\mathbf{x}}_{N|k} \in \mathcal{S}_{t_\text{red}}$ is to ensure that the vehicle stays behind the traffic light until $t_\text{red}+N$ steps, while the constraint $\hat{\mathbf{x}}_{N|k} \in \mathcal{P}_{t_\text{green}}$ is to ensure that the vehicle passes the traffic light within $t_\text{green}+N$ steps.}
\label{fig:planning_strategy}
\end{center}
% %\vspace{-1.0em}
\end{figure}

In this paper, the sets $\mathcal{S}_{t_\text{red}}$ and $\mathcal{P}_{t_\text{green}}$ are calculated in a data-driven way by Algorithm \ref{alg: RNTW}.
The data $\mathcal{D}$ in this algorithm are feasible state-input pairs $\{(\mathbf{x}_d, ~u_d)\}_{d=0}^{N_\text{data}}$, meaning that $\mathbf{x}_d \in \mathcal{X}$ and $u_d \in \mathcal{U}$ for all $d \in \{0,\cdots, N_\text{data}\}$. 
$\mathcal{R}_{i-1}$ represents $i-1$-step robust controllable set, which is initialized with to the target set $\mathcal{T}$ at line 2.
Line $6$ at the iteration $i$ checks whether the state $\mathbf{x}_d$ of the system \eqref{eq: observer equation} can be robustly steered to $\mathcal{R}_{i-1}$. If the condition is satisfied, the state $\mathbf{x}_d$ is included in a set $\mathbb{X}$. At the end of the iteration $i$ (line 9), $\mathcal{R}_i$ is set as the convex hull of the set $\mathbb{X}$.
% %\vspace{-1.0em}
\begin{algorithm} \label{alg: RNTW}
    % \SetKwFunction{isMember}{isMember}
    \SetKwInOut{KwIn}{Input}
    \SetKwInOut{KwOut}{Output}
    \KwIn{Data of feasible pairs $\mathcal{D} = \{(\mathbf{x}_d, ~u_d)\}_{d=0}^{N_\text{data}}$, \\ Remaining steps $t$, Target convex set $\mathcal{T}$, \\ System parameters $\mathbf{A}, ~\mathbf{B}, ~\mathbf{F}, ~\mathcal{W}, ~L$}
    \KwOut{$t$-step \textit{Robust Controllable Set} $\mathcal{R}_t$} %, \\ The associated data point indices}

    $\mathcal{R}_0 \leftarrow \mathcal{T}$
    
    % $\mathcal{I} = \emptyset$
    
    % \tcc{For odd elements in the list, we add 1, and for even elements, we add 2.
    % After the loop, all elements are even.}
    \For{$i \leftarrow 1$ \KwTo $t$}{
        $\mathbb{X} \leftarrow \emptyset$
        
        \For{$j \leftarrow 0$ \KwTo $N_\text{data}$}{
            \If{$\mathbf{A}\mathbf{{x}}_j + \mathbf{B}u_j \in \mathcal{R}_{i-1} \ominus 2L \mathbf{F} \mathcal{W}$}{
                $\mathbb{X} \leftarrow \mathbb{X} \medcup \{\mathbf{{x}}_j\}$
                
                % \If{i = t}{
                %     $\mathcal{I} \leftarrow \mathcal{I} \medcap \{j\}$
                % }
             }
        }
        $\mathcal{R}_i \leftarrow \text{conv}(\mathbb{X})$ \\
    }
    \KwRet{$\mathcal{R}_t$} % and $\mathcal{I}$}
    \caption{Data-driven computation of \\ $~~~~~~~~~~~~~~~~~~ t$-step \textit{Robust Controllable Set} $\mathcal{R}_t$}
\end{algorithm} \\
% %\vspace{-1.0em}\\
% \noindent Algorithm \ref{alg: RNTW} generates the set sequence $\{\mathcal{R}_i\}$ satisfying $\mathcal{R}_{i-1} \subseteq \mathcal{R}_{i}$ for $i \in \{1,\cdots,t\}$ because $\mathbf{{x}}_j \in \mathcal{R}_{i-1}$ automatically satisfies $\mathbf{A}\mathbf{{x}}_j + \mathbf{B}u_j \in \mathcal{R}_{i-1} \ominus \mathbf{F} \mathcal{W}$ by construction.
By the following theorem, the output of the Algorithm \ref{alg: RNTW} is a $t$-step robust controllable set for a given target convex set $\mathcal{T}$. 
\begin{theorem} \label{thm1}
    % Suppose that the set $\mathcal{R}_t$ calculated by Algorithm \ref{alg: RNTW} is not empty.
    The output of Algorithm \ref{alg: RNTW} $\mathcal{R}_t$ is $t$-step robust controllable set of the system \eqref{eq: observer equation} perturbed by the noise \eqref{eq: lumped noise} for a given target convex set $\mathcal{T}$ subject to the constraints \eqref{eq:constraints}.  
\end{theorem}
%\vspace{-1.0em}
\begin{proof}
    See Appendix.
\end{proof}
% \begin{remark}
%     The robust controllable sets $\mathcal{S}_{t_\text{red}}$ and $\mathcal{P}_{t_\text{green}}$ can be calculated as in \cite[Sec. 10.3.2]{borrelli2017predictive}.
%     However, $\mathcal{S}_{t_\text{red}}$ and $\mathcal{P}_{t_\text{green}}$ are not only the robust controllable sets but also the domain of the terminal cost $V(\cdot)$. As the domain of $V(\cdot)$ will be a convex hull of the states in data set $\mathcal{D} = \{(\mathbf{x}_d, ~u_d)\}_{d=0}^{N_\text{data}}$, we use Algorithm \ref{alg: RNTW} to calculate $\mathcal{S}_{t_\text{red}}$ and $\mathcal{P}_{t_\text{green}}$.
% \end{remark}

\subsection{Learning the terminal cost $V(\cdot)$} \label{sec: learning V function}
The terminal cost $V(\cdot)$ of the proposed MPC in Sec. \ref{sec: Data-driven MPC} represents the cost-to-go (or remaining energy-to-spend) function. 
% We will show that the terminal cost $V(\cdot)$ monotonically decreases as the amount of the data increases.
The classical way to solve this problem is dynamic programming \cite{bertsekas2012dynamic, bae2022ecological, yang2020convex} which requires gridding the state and input space.
Due to this gridding, the control input is restricted to elements of the discrete inputs, which is not desirable for automated vehicles where the passengers’ comfort is also a major concern \cite{gonzalez2015review}.
Moreover, as the size of the grid decreases, its computational cost exponentially increases.
To resolve this issue, we modify and employ a data-driven method in \cite{ejoalmpc}, which does not require gridding and provides continuous inputs.

Given the data set $\mathcal{D}$, we calculate the cost-to-go values for all data points in backward from the set $\mathcal{T} = \{\mathbf{x} ~|~ \begin{bmatrix} 1 & 0 \end{bmatrix} \mathbf{x} \geq s_\text{tl}\}$, where the cost-to-go value of every point is zero. Subsequently, the terminal cost $V(\cdot)$ is calculated as a convex combination of the calculated cost-to-go values. Specifically, Algorithm \ref{alg: ctg value for data pts} is employed to calculate the terminal cost. Note that $\mathbf{J}_k$ declared in line $5$ contains the cost-to-go values of the data points $\{(\mathbf{x}_d, ~u_d)\}_{d=0}^{N_\text{data}}$ calculated until $k$ iterations of while loop. $\bm{\Delta}$ in line $14$ is a set of convex coefficients, i.e., $\bm{\Delta}=\{\bm{\lambda} \in \mathbb{R}^{N_\text{data} \times 1} ~|~ \bm{\lambda} \geq \mathbf{0}^{N_\text{data} \times 1}, ~ \bm{1}^\top \bm{\lambda} = 1\}$ where $\bm{1}$ is an $N_\text{data} \times 1$ vector filled with ones. The expectation in line $8$ is approximated with a sample mean of its value.
% %\vspace{-1.0em}
\begin{algorithm} \label{alg: ctg value for data pts}
    % \SetKwFunction{isMember}{isMember}
    \SetKwInOut{KwIn}{Input}
    \SetKwInOut{KwOut}{Output}
    \KwIn{Data of feasible pairs $\mathcal{D} = \{(\mathbf{x}_d, ~u_d)\}_{d=0}^{N_\text{data}}$, \\ Target convex set $\mathcal{T} = \{\mathbf{x} ~|~ \begin{bmatrix} 1 & 0 \end{bmatrix} \mathbf{x} \geq s_\text{tl}\}$, \\
    Parameters $\mathbf{A}, ~\mathbf{B}, ~\mathbf{F}, ~\mathbf{W}, ~L$, $q(n)$}
    \KwOut{The cost-to-go function $V(x)$}
  
    $V_0(\mathbf{{x}}) = 0 ~\forall \mathbf{{x}} \in \mathcal{T}$

    $\mathbf{J}_0 \leftarrow \mathbf{0}^{1 \times N_\text{data}}, ~ k \leftarrow -1$

    \DontPrintSemicolon
    \Repeat{$\mathbf{J}_{k+1} = \mathbf{J}_k$}{
        $\mathbf{J}_{k+1} \leftarrow \mathbf{0}^{1 \times N_\text{data}}, ~k \leftarrow k+1$
        
        $\mathcal{R}_{k} \leftarrow$ Algorithm \ref{alg: RNTW}$(\mathcal{D}, k, \mathcal{T})$
        
        \For{$j \leftarrow 0$ \KwTo $N_\text{data}$}{
            \eIf{$\mathbf{A}\mathbf{{x}}_j + \mathbf{B}u_j \in \mathcal{R}_{k} \ominus 2L\mathbf{F} \mathcal{W}$}{
                $[\mathbf{J}_{k+1}]_j \leftarrow \ell(\mathbf{{x}}_j, u_j) + \mathbb{E}_n[V_k(\mathbf{A}\mathbf{{x}}_j+ \mathbf{B}u_j + \mathbf{F}n)]$
             }
             {$[\mathbf{J}_{k+1}]_j \leftarrow \infty$}
        }

        $V_{k+1}(\mathbf{x}) \leftarrow \min_{\bm{\lambda} \in \bm{\Delta}} \mathbf{J}_k \bm{\lambda}$ s.t. $\sum_{d=0}^{N_\text{data}} [\bm{\lambda}]_d \mathbf{x}_d  = \mathbf{x}$
    }
    $V(\mathbf{x}) \leftarrow V_{k+1}(\mathbf{x})$ 
    
    \KwRet{$V(\mathbf{x})$}
    \caption{Data-driven computation of $V(\cdot)$}
\end{algorithm}
% %\vspace{-1.5em}

\subsection{Model Predictive Control for Eco-driving} \label{sec: Data-driven MPC}
By incorporating the terminal constraints and the terminal cost function, we design an MPC controller as follows:
%\vspace{-0.12cm}
\begin{equation} \label{eq:mpc ori}
\begin{aligned}
    & J_\text{MPC}(\hat{\mathbf{x}}_k) =  \\
    & \min_{\Pi(\cdot)} ~\,\mathbb{E}_{n_{0:N-1}}\Biggr[\sum_{i=0}^{N-1} \ell(\hat{\mathbf{x}}_{i|k}, ~ \pi_k(\hat{\mathbf{x}}_{i|k})) + V(\hat{\mathbf{x}}_{N|k}) \Biggr] \\
    & ~ \textnormal{s.t.,} ~~ \hat{\mathbf{x}}_{i+1|k} =\mathbf{A} \hat{\mathbf{x}}_{i|k} + \mathbf{B} \pi_k(\hat{\mathbf{x}}_{i|k}) + \mathbf{F}n_i, \\
    % & \qquad   ~\, n_i = L\Delta s_{i|k} + Lw_i, \\
    % & \qquad ~\,\Delta s_{i+1|k} = (1-L)\Delta s_{i|k} - Lw_i, \\
    & \qquad ~\, \hat{\mathbf{x}}_{0|k}=\hat{\mathbf{x}}_k, ~n_i \sim q(n), \\
    & \qquad ~\, \hat{\mathbf{x}}_{1|k} \oplus \mathbf{F}\mathcal{W} \subseteq \mathcal{T}_s,  ~\text{if the light is red,} \\
    & \qquad ~\, \hat{\mathbf{x}}_{i|k} \in \mathcal{X}, ~\pi_k(\hat{\mathbf{x}}_{i|k})) \in \mathcal{U}, ~\forall n_i \in 2L\mathcal{W},, \\
    & \qquad ~\, \hat{\mathbf{x}}_{N|k} \oplus \mathbf{F}\mathcal{W} \subseteq \mathcal{S}_{t_\text{red}} \medcap \mathcal{P}_{t_\text{green}}, ~\forall n_i \in 2L\mathcal{W},\\
    & \qquad ~\, i=0,...,N-1 
\end{aligned}
%\vspace{-0.12cm}
\end{equation}
where $\hat{\mathbf{x}}_{i|k}$ is a predicted state estimate at time step $k+i$, and $\mathcal{T}_s = \{\mathbf{x} ~|~ \begin{bmatrix} 1 & 0 \end{bmatrix} \mathbf{x} \leq s_\text{tl}\}$.
Note that the estimation error $\Delta s \in \mathcal{W}, ~\forall n_k \in 2L\mathcal{W}$ from Proposition \ref{prop: delta s and lumped noise}.
Thus, the terminal constraint in \eqref{eq:mpc ori} ensures that the predicted actual state ${\mathbf{x}}_{N|k} \in \mathcal{S}_{t_\text{red}} \medcap \mathcal{P}_{t_\text{green}}$, and the constraint $\hat{\mathbf{x}}_{1|k} \oplus \mathbf{F}\mathcal{W} \subseteq \mathcal{T}_s$ ensures the vehicle stay behind the traffic light if traffic light color is red.
% The optimal control problem \eqref{eq:mpc ori} minimizes the expected sum of the regressed energy consumption state cost $\ell(\cdot,\cdot)$ in \eqref{eq: energy cost regression} evaluated for the estimated state trajectory. By \eqref{eq: CE stage cost}, this cost is identical to the cost evaluated for the actual state trajectory.
We point out that as the system \eqref{eq: observer equation} is uncertain, the optimal control problem \eqref{eq:mpc ori} consists of finding state feedback policies \(\Pi(\cdot) = \{\pi_0(\cdot),\pi_1(\cdot),...,\pi_{N-1}(\cdot)\} \).
This MPC formulation is not computationally tractable as optimizing over control policies \( \{\pi_0(\cdot),\pi_1 (\cdot),...\} \) involves an infinite-dimensional optimization.

We simplify the control policy as \(\pi_k(\cdot) = u_k\), eliminating the need for expectations on the stage costs, and approximate the expectation of the terminal cost as a sample mean. Additionally, we reformulate the constraints on estimate states to those on nominal states\footnote{See Appendix for the details of the constraint reformulation and sample mean approximation.}.
Specifically, we design a tractable MPC controller as follows:
%\vspace{-0.12cm}
\begin{equation} \label{eq:mpc reform}
\begin{aligned}
    & \hat{J}_\text{MPC} (\hat{\mathbf{x}}_{k}, ~sn_{1:M}) = \\
    & \min_{\substack{u_{0:N-1|k}}} ~\, \sum_{i=0}^{N-1} \ell(\bar{\mathbf{x}}_{i|k}, ~ u_{i|k}) +\frac{1}{M}\sum_{m=1}^{M} V(\bar{\mathbf{x}}_{N|k} + \mathbf{F}sn_m)\\
    & ~~~~ \textnormal{s.t.,} \, ~~~~ \bar{\mathbf{x}}_{i+1|k} =\mathbf{A} \bar{\mathbf{x}}_{i|k} + \mathbf{B} u_{i|k} \\
    & \qquad \qquad ~ \bar{\mathbf{x}}_{0|k} = \hat{\mathbf{x}}_{k}, \\
    & \qquad \qquad ~ \bar{\mathbf{x}}_{1|k} \in \mathcal{T}_s \ominus (2L+1)\mathbf{F}\mathcal{W},  ~\text{if the light is red,} \\
    & \qquad \qquad ~ \bar{\mathbf{x}}_{i|k} \in \mathcal{X}, ~ u_{i|k} \in \mathcal{U},\\
    & \qquad \qquad ~ \bar{\mathbf{x}}_{N|k} \in \mathcal{S}_{t_\text{red}} \ominus (2LN+1)\mathbf{F}\mathcal{W},\\
     & \qquad \qquad ~ \bar{\mathbf{x}}_{N|k} \in \mathcal{P}_{t_\text{green}} \ominus (2LN+1)\mathbf{F}\mathcal{W},
\end{aligned}
%\vspace{-0.12cm}
\end{equation}
where $sn_{1:M}$ denotes $M$ samples of noise derived from the random variable $\sum_{i=0}^{N-1}n_i$, where $n_i$ follows a distribution $q(w)$, $\bar{\mathbf{x}}_{i|k}$ is the nominal state and $u_{i|k}$ is the input at predicted time step $k+i$. 
Note that the objective function in \eqref{eq:mpc reform} is identical to that in \eqref{eq:mpc ori} since the stage cost in \eqref{eq: energy cost regression} does not depend on the position and thus $\ell(\hat{\mathbf{x}}_k, u_k) = \ell(\bar{\mathbf{x}}_k, u_k)$.

\subsection{Implementation}
The proposed MPC \eqref{eq:mpc reform} is a convex optimization problem.
The stage cost is convex as $\mathbf{P} \succeq 0$ \eqref{eq: energy cost regression}. The terminal cost $V(\cdot)$ is convex as it is the optimal objective function of a multiparametric linear program \cite[Thm 6.5]{borrelli2017predictive} as described in line 14 of Algorithm \ref{alg: ctg value for data pts}.
% Thus, the cost function is a sum of the convex functions, which is convex \cite{boyd2004convex}.
The system equation is linear, and state/input constraints \eqref{eq:constraints} are convex. 
The set $\mathcal{W}$ is convex polytope \eqref{eq: gps error}, and the sets $\mathcal{S}_{t_\text{red}}$ and $\mathcal{P}_{t_\text{green}}$ are convex polytopes as they are calculated via convex hull operation of a finite number of points as described in line 9 of Algorithm \ref{alg: RNTW}. Since the Pontryagin difference between two convex polytopes is convex, the terminal constraints are convex.

To implement the proposed MPC \eqref{eq:mpc reform}, we utilize CVXPY \cite{diamond2016cvxpy} as a modeling language and MOSEK as a solver \cite{mosek}.

\section{Validation} \label{sec: experiment}
\subsection{Baseline}
There are two baseline algorithms considered in this study. 
The first baseline algorithm is the cruise controller detailed in Section \ref{sec: data}. 
The second baseline algorithm is the algorithm introduced in \cite{bae2019VIL}.
The major differences between the algorithm in \cite{bae2019VIL} and the proposed algorithm are:
\begin{itemize}
    \item The algorithm in \cite{bae2019VIL} disregards localization uncertainty, while the proposed algorithm addresses this uncertainty.
    \item The algorithm in \cite{bae2019VIL} cannot ensure passing each traffic light within a specified time, whereas the proposed algorithm can do so through terminal constraints.
    \item The algorithm in \cite{bae2019VIL} determines speed references for the entire route, while the proposed algorithm calculates speed references for a short prediction horizon, requiring less computational power.
    \item The algorithm in \cite{bae2019VIL} uses an additional speed tracking controller, while the proposed algorithm, as an MPC controller with longitudinal acceleration input, does not.
\end{itemize}
Note that for three algorithms, there is an actuator-level controller that converts longitudinal acceleration commands into electric motor torque inputs.

\subsection{Localization Uncertainty and Controller Parameters}
Throughout the validation, we assume that a localization module provides position information every $1 sec$ under the following localization uncertainty $w$:
\begin{equation}
    p(w) \sim U(-3, 3), ~\mathcal{W}=\{w ~|~ 3\leq w \leq 3\},
\end{equation}
where $U$ represents a uniform distribution. $3m$ in longitudinal position error exceeds the accuracy of lane-level positioning \cite{williams2020qualitative}, which can be achievable using off-the-shelf systems using vision \cite{mobileye} and V2I communication \cite{cohda}. It is $95 \%$ Circle of Error Probable (CEP) for commercially available GPS \cite{vbox}.

We set the prediction horizon $N=5$, which represents $5$sec prediction as the discretization time $T_s = 1$sec.
The observer gain $L$ in \eqref{eq: observer equation} is set to $\frac{1}{4N}$. 

\subsection{Data collection and Training} \label{sec: single road segment}
We collect the data from two simple scenarios using simulation. Given the simplicity of the system \eqref{eq:system}, data collection through simulation is feasible.
% The intuition behind this approach is to excel at simple tasks before evaluating performance in complex, realistic scenarios.

In the first scenario, the vehicle's initial speed is set to zero and the testing road is a single road segment with fixed parameters outlined in Table \ref{table: param 1}. Essentially, the vehicle needs to pass an upcoming traffic light, situated $200$m away, within a 20-second interval as directed by the high-level module.
%\vspace{-1.0em}
\begin{table}[h]
\centering
\caption{Parameters for the single road segment scenario}
%\vspace{-2mm}
\label{table: param 1}
\begin{tabular}{ ||l||c||} 
 \hline
 Parameter & Value\\
 \hline
 Current traffic signal &  Green \\
 Traffic light cycle  & Green: $30$s, Yellow: $5$s, Red: $25$s \\ 
 Remaining time/distance & $25$s/$200$m \\
 Time to pass traffic light & $20$s \\
 \hline
\end{tabular}
\end{table}
%\vspace{-1.0em}\\

As explained in Sec. \ref{sec: data}, we initialize the data set by executing the cruise controller multiple times at different speeds. Subsequently, we employ MPC \eqref{eq:mpc reform}, computing the terminal constraints from Algorithm \ref{alg: RNTW} and the terminal cost from Algorithm \ref{alg: ctg value for data pts}. While running the MPC \eqref{eq:mpc reform}, we record the closed-loop state and input pairs. After completing each task iteration, we augment the data using these state-input pairs.

Fig. \ref{fig:energyconsumption_single} illustrates the learning curve in the single road segment scenario with fixed parameters. 
The x-axis denotes the number of the data augmentation process by iteratively conducting the given scenario.
Due to the uncertainty described in system \eqref{eq:system}, the total energy consumption is a random variable. Therefore, we calculate the sample mean of the total energy consumption by performing 100 Monte Carlo simulations of the system \eqref{eq:system} in closed-loop using the MPC \eqref{eq:mpc reform} after each data augmentation process is completed. 
The results demonstrate $16 \%$ improvement in total energy consumption as the number of data augmentations (or task iterations) increases, eventually leading to settled performance.
% %\vspace{-1.0em}
\begin{figure}[ht]
\begin{center}
\includegraphics[width=1\linewidth,keepaspectratio]{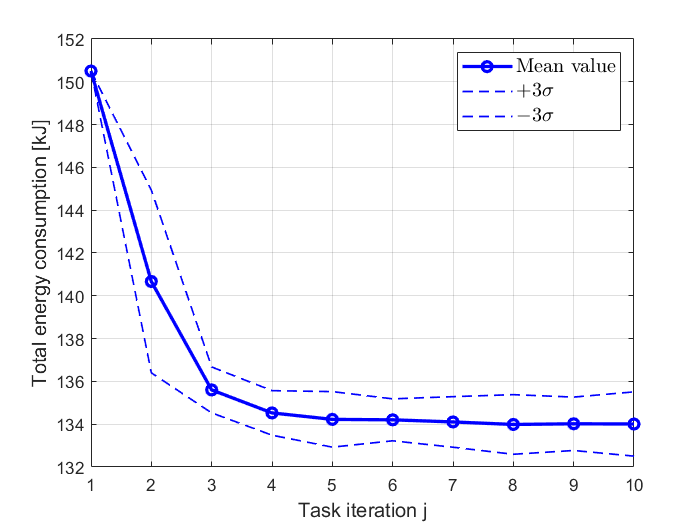}
%\vspace{-1.0em}
\caption{Total Energy Consumption Improvement with Increasing Data Size. $100$ Monte Carlo simulations for each task iteration.}
\label{fig:energyconsumption_single}
\end{center}
\end{figure}
% %\vspace{-1.0em}\\

In the second scenario, the vehicle maneuvers on a single road segment with randomized parameters and random initial speed.
To cope with realistic scenarios where the parameter values and the initial speed vary, we introduce this training process. 
Upon augmenting the data $10$ times in this second scenario, we conclude the training process.

\subsection{Test setup and scenario}
Our test setup and scenario are presented in Fig. \ref{fig:test}. 
\begin{figure}
\vspace{+1.0em}
\begin{center}
\includegraphics[width=1.0\linewidth,keepaspectratio]{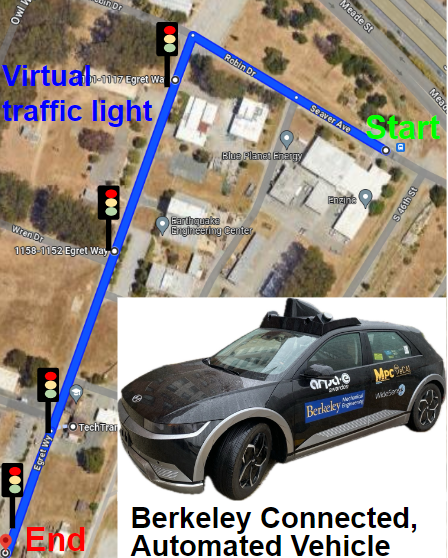}
%\vspace{-1.0em}
\caption{Test vehicle and scenario}
\label{fig:test}
\end{center}
%\vspace{-1.0em}
\end{figure}\\
We utilized the retrofitted Hyundai Ioniq 5 as the test vehicle maneuvering around the physical testing site. Additionally, along the testing route, we located four virtual traffic lights at distances $s=[189, \text{m}, 378, \text{m}, 490, \text{m}, 553, \text{m}]$. We designed the traffic cycle of four traffic lights such that the vehicle can pass without stopping if it maintains a constant speed of $5 \text{m/s}$. The specified times for passing each traffic light are $k_\text{pass}=[43 , \text{sec}, 81 \text{sec}, 103 \text{sec}, 116 \text{sec}]$. There are no surrounding vehicles to evaluate the performance under free-flow conditions.

The initial condition of the test vehicle is idle speed, meaning no driver pedal inputs. Starting from idle speed, the vehicle needs to complete the given route under localization uncertainty while crossing each traffic light within the time specified by the high-level green wave search module.

\subsection{Test results}
The reference speed of the cruise controller is set to finish the given route within $116sec$, which is the specified time to pass the fourth traffic light.
The regularization parameter in \cite[eq. (4)]{bae2019VIL} is tuned to show a similar travel time to the proposed algorithm.
For each algorithm, we conduct the same scenario five times and evaluate the energy consumption in \eqref{eq: energy definition}.
The change of $E_\text{stor}$ is calculated using the equipped voltage and current sensors, while the change of $E_\text{kin}$ is calculated by a change in kinetic energy between the initial and the end states.
Energy consumption results are described in Table \ref{table: energy eval}.
Compared to the cruise control with constant speed reference, the proposed algorithm shows $22\%$ improvement in average energy consumption. Compared to the algorithm in \cite{bae2019VIL}, the proposed algorithm shows $11.6\%$ improvement in average energy consumption while the proposed algorithm is $14.9\%$ faster in average travel time.
%\vspace{-1.0em}
\begin{table}[h]
\centering
\caption{Closed-loop performance of algorithms. Mean(Minimum/Maximum)}
%\vspace{-2mm}
\label{table: energy eval}
\begin{tabular}{ ||l||c||c||} 
 \hline
 Algorithm&Energy consumption[kJ]&Travel time[s]\\
 \hline
 Cruise control&$260.0 (239.8/285.0)$&$114.5(114.1/115.0)$\\
 Algorithm \cite{bae2019VIL} &$229.4(202.3/279.8)$ &$133.2(132.0/136.1)$\\ 
 Proposed&$202.7(187.4/215.1)$&$113.4(113.1/113.5)$\\
 \hline
\end{tabular}
\end{table}
%\vspace{-0.5em}\\

% \begin{table}[h]
% \centering
% \caption{Closed-loop performance of algorithms. Mean(Minimum/Maximum)}
% %\vspace{-2mm}
% \label{table: energy eval}
% \begin{tabular}{ ||l||c||c||} 
%  \hline
%  Algorithm&Energy consumption[kJ]&Travel time[s]\\
%  \hline
%  Cruise control&$260.0 (239.8/285.0)$&$114.5(114.1/115.0)$\\
%  State-of-the-art &$229.4(202.3/279.8)$ &$133.2(132.0/136.1)$\\ 
%  Proposed&$202.7(187.4/215.1)$&$113.4(113.1/113.5)$\\
%  \hline
% \end{tabular}
% \end{table}
% %\vspace{-0.5em}\\

The speed and the travel time profiles are presented in Fig.\ref{fig:energyconsumption_test}.
As illustrated in Fig.\ref{fig:energyconsumption_test}. (b), the proposed algorithm crosses each traffic light within the time specified by a high-level green wave search module.
This results in the completion of the given route with small fluctuation as illustrated in Fig.\ref{fig:energyconsumption_test}. (a).
In contrast, other algorithms show larger fluctuation as the controlled vehicles encounter red light.
Though the algorithm \cite{bae2019VIL} ensures timely crossing of traffic lights during the planning phase, the tracking controller fails to follow the planned speed reference due to tracking error, time latency, and difference in objective.
%\vspace{-1.0em}
\begin{figure}[ht]
\begin{center}
\includegraphics[width=1\linewidth,keepaspectratio]{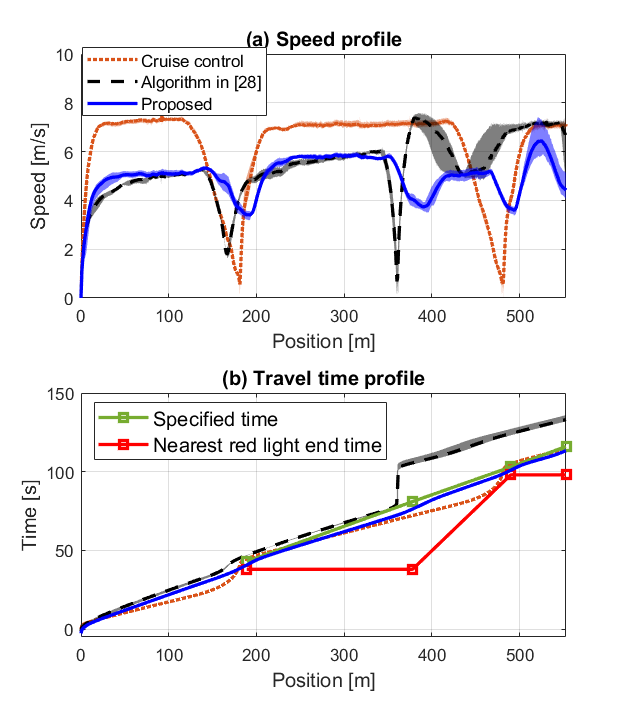}
%\vspace{-3.0em}
\caption{Comparison: The proposed algorithm and the algorithm \cite{bae2019VIL}}
\label{fig:energyconsumption_test}
\end{center}
\end{figure}\\

We further validated our regressed model in \eqref{eq: energy cost regression} with the test data. For each trial, we calculated the error in total energy consumption between the measurements and the regression model \eqref{eq: energy cost regression} as a percentage. The results are presented in Table \ref{table: energy eval2}. Note that the value of the worst case error ($6.3 \%$) is lower than the minimum energy saving in Table \ref{table: energy eval}.
\begin{table}[h]
\centering
\caption{Error of the total energy consumption (\%)}
%\vspace{-2mm}
\label{table: energy eval2}
\begin{tabular}{ ||l||c||c||} 
 \hline
 Mean[\%]&Standard deviation[\%]&Worst case[\%]\\
 \hline
 $-1.06\%$&$3.41 \%$&$-6.3\%$\\
 \hline
\end{tabular}
\end{table}
%\vspace{-0.5em}\\

% \section{Conclusion and Future work}
% The novel eco algorithm for CAVs is presented.

\section*{ACKNOWLEDGMENT}
This research work presented herein is funded
by the Advanced Research Projects Agency-Energy 
(ARPA-E), U.S. Department of Energy under DE-AR0000791. 
\bibliography{IEEEabrv,reference.bib}
%\vspace{-0.12cm}
\section*{Appendix}
%\vspace{-0.12cm}
\subsection{Proof of Theorem \ref{thm1}}
%\vspace{-0.12cm}
All data points in $\mathcal{D}$ satisfy the constraints \eqref{eq:constraints} by construction. We prove the rest of the claim by induction

    For $i=0$, by definition of the robust controllable set \cite[Def.10.18]{borrelli2017predictive}, $\mathcal{R}_0 =\mathcal{T}$ is a $0$-step robust controllable set of the system \eqref{eq: observer equation} perturbed by the noise \eqref{eq: lumped noise} for the target convex set $\mathcal{T}$ subject to the constraints \eqref{eq:constraints}.
    
    Suppose that for some $i \geq 0$, $\mathcal{R}_i$, declared in line 9 at iteration $i$, is an $i$-step robust controllable set of the system \eqref{eq: observer equation} perturbed by the noise \eqref{eq: lumped noise} for the target convex set $\mathcal{T}$ subject to the constraints \eqref{eq:constraints}.
    In line 6, we find data points such that the state $\mathbf{x}_j$ of the system \eqref{eq: observer equation} can be robustly steered to $\mathcal{R}_i$. 
    Let $\{(\mathbf{\Tilde{x}}_s, ~\Tilde{u}_s)\}_{s=0}^{N_\text{s}}$ denote a set of the data points that satisfy the condition in line 6. 
    By definitions \cite[Def.10.15 \& 18]{borrelli2017predictive}, $\{\mathbf{\Tilde{x}}_s\}_{s=0}^{N_\text{s}}$ are elements of a $i+1$-step robust controllable set. 
    
    Now, we prove convex combinations of $\{\mathbf{\Tilde{x}}_s\}_{s=0}^{N_\text{s}}$ are also elements of the $i+1$-step robust controllable set.
    First, we show that convex combinations of data points $\{(\mathbf{\Tilde{x}}_s, ~\Tilde{u}_s)\}_{s=0}^{N_\text{s}}$ satisfy the constraints \eqref{eq:constraints}. This is true because $\mathcal{X}$ and $\mathcal{U}$ are convex sets.
    Second, we prove that convex combinations of $\{\mathbf{\Tilde{x}}_s\}_{s=0}^{N_\text{s}}$ can be robustly steered to $\mathcal{R}_i$.
    $\mathcal{R}_i$ is convex as it is constructed by the convex hull operation in line 9 at iteration $i$. The following holds for all $0 \leq \lambda \leq 1$:
    %\vspace{-0.12cm}
    \begin{equation*}
        \begin{aligned}
            & ~~~~~~~~\forall s_1, s_2 \in \{0,\cdots,N_\text{s}\},\\
            & ~~~~~~~~\mathbf{A}\mathbf{\Tilde{x}}_{s_1} + \mathbf{B}\Tilde{u}_{s_1} + \mathbf{F}n_{s_1} \in \mathcal{R}_i, \forall n_{s_1} \in 2L\mathcal{W},\\
            & ~~~~~~~~\mathbf{A}\mathbf{\Tilde{x}}_{s_2} + \mathbf{B}\Tilde{u}_{s_2} + \mathbf{F}n_{s_2} \in \mathcal{R}_i, \forall n_{s_2} \in 2L\mathcal{W}.\\
            & \implies \forall s_1, s_2 \in \{0,\cdots,N_\text{s}\}, \forall n_{s_1}, n_{s_2}  \in 2L\mathcal{W}, \\
            & ~~~~~~~~ \mathbf{\Tilde{x}}_c \coloneqq \lambda\mathbf{\Tilde{x}}_{s_1} + (1-\lambda)\mathbf{\Tilde{x}}_{s_2},\\
            & ~~~~~~~~ \mathbf{\Tilde{u}}_c \coloneqq \lambda\mathbf{\Tilde{u}}_{s_1} + (1-\lambda)\mathbf{\Tilde{u}}_{s_2}, \\
            & ~~~~~~~~ n_{c} \coloneqq \lambda n_{s_1} + (1-\lambda)n_{s_2}, \\
            & ~~~~~~~~ \mathbf{A}\mathbf{\Tilde{x}}_c + \mathbf{B}\mathbf{\Tilde{u}}_c + \mathbf{F}n_{c} \in \mathcal{R}_i ~(\because \mathcal{R}_i \text{ is convex.}).
        \end{aligned}
    %\vspace{-0.12cm}
    \end{equation*}
    Moreover, as $\mathcal{W}$ is a convex set, $n_{c}$ belongs to the set $2L\mathcal{W}$ and can represent all realizations of noise in $2L\mathcal{W}$.
    Thus, all convex combinations of $\{\mathbf{\Tilde{x}}_s\}_{s=0}^{N_\text{s}}$ are also elements of the $i+1$-step robust controllable set. Therefore, $\mathcal{R}_{i+1}$ is an $i+1$-step robust controllable set when $\mathcal{R}_i$ is an $i$-step robust controllable set.
    By induction, the claim is proved.

%\vspace{-0.12cm}
\subsection{Constraint reformulation}
%\vspace{-0.12cm}
First, we consider the state and input constraints \eqref{eq:constraints}. These constraints do not depend on the noisy position. Moreover, the speed of the nominal state propagated from $\bar{\mathbf{x}}_{0|k}=\hat{\mathbf{x}}_k$ is identical to that of the estimated state. Thus, we have that:
%\vspace{-0.12cm}
\begin{equation*}
\begin{aligned}
    & \bar{\mathbf{x}}_{i|k} \in \mathcal{X}, ~ u_{i|k} \in \mathcal{U} \implies \hat{\mathbf{x}}_{i|k} \in \mathcal{X}, ~ u_k \in \mathcal{U}.
\end{aligned}
%\vspace{-0.12cm}
\end{equation*}

Second, we reformulate the terminal constraints in \eqref{eq:mpc ori}.
From the dynamics in \eqref{eq:mpc ori}, the nominal dynamics in \eqref{eq:mpc reform}, and \eqref{eq: lumped noise}, we have that:
%\vspace{-0.12cm}
\begin{equation*}
\begin{aligned}
    & \mathbf{e}_{i|k} = \hat{\mathbf{x}}_{i|k} - \bar{\mathbf{x}}_{i|k}, \\
    & \mathbf{e}_{i+1|k} = \mathbf{A}\mathbf{e}_{i|k} + \mathbf{F}n_i, ~n_i \sim q(n).
\end{aligned}
%\vspace{-0.12cm}
\end{equation*}
From the initial constraint in \eqref{eq:mpc reform}, $\mathbf{e}_{0|k} = 0$.
Moreover, $\mathbf{A}^i \mathbf{F} = \mathbf{F}$.
Thus, we have that:
%\vspace{-0.12cm}
\begin{equation} \label{eq: error nominal estimate}
\begin{aligned}
    & \mathbf{e}_{N|k} = \sum_{i=0}^{N-1} \mathbf{A}^i \mathbf{F}n_i = \mathbf{F} \sum_{i=0}^{N-1}n_i.
\end{aligned}
%\vspace{-0.12cm}
\end{equation}
From \eqref{eq: lumped noise}, we have that $\mathbf{e}_{N|k} \in 2LN\mathbf{F} \mathcal{W}$ which implies:
%\vspace{-0.12cm}
\begin{equation*}
\begin{aligned}
    & \hat{\mathbf{x}}_{N|k} \in  \bar{\mathbf{x}}_{N|k} \oplus 2LN\mathbf{F}\mathcal{W}, ~\forall n_k \in 2L\mathcal{W}.
\end{aligned}
%\vspace{-0.12cm}
\end{equation*}
Then, we have the following:
%\vspace{-0.12cm}
\begin{equation*}
\begin{aligned}
    & \hat{\mathbf{x}}_{N|k} + \mathbf{F}\mathcal{W} \in \bar{\mathbf{x}}_{N|k} \oplus 2LN\mathbf{F}\mathcal{W} \oplus \mathbf{F} \mathcal{W}, ~\forall n_k \in 2L\mathcal{W}.
\end{aligned}
%\vspace{-0.12cm}
\end{equation*}
Thus, the terminal constraint on the nominal state in \eqref{eq:mpc reform} is a necessary condition to the terminal constraint in \eqref{eq:mpc ori}.

Third, we reformulate the red light constraint in \eqref{eq:mpc ori}.
Similar to the terminal constraint reformulation, we have that:
%\vspace{-0.12cm}
\begin{equation}
\begin{aligned}
    & \hat{\mathbf{x}}_{1|k} \oplus \mathbf{F}\mathcal{W} \in  \bar{\mathbf{x}}_{1|k} \oplus 2L\mathbf{F}\mathcal{W} \oplus \mathbf{F}\mathcal{W}  \subseteq \mathcal{T}_s.
\end{aligned}
%\vspace{-0.12cm}
\end{equation}
Thus, the red light constraint on the nominal state in \eqref{eq:mpc reform} is a necessary condition to the red light constraint in \eqref{eq:mpc ori}.
%\vspace{-0.12cm}
\subsection{Sample Mean Approximation}
%\vspace{-0.12cm}
From \eqref{eq: error nominal estimate}, the expectation of the terminal cost in \eqref{eq:mpc ori} can be reformulated as follows:
%\vspace{-0.12cm}
\begin{equation*}
\begin{aligned}
    & \mathbb{E}_{w_{0:N-1}}[V(\hat{\mathbf{x}}_{N|k})] = \mathbb{E}_{n_{0:N-1}}\bigg[V\bigg(\bar{\mathbf{x}}_{N|k} + \mathbf{F} \sum_{i=0}^{N-1}n_i\bigg)\bigg].
\end{aligned}
%\vspace{-0.12cm}
\end{equation*}
Considering that $sn_{1:M}$ denotes $M$ samples of noise derived from the random variable $\sum_{i=0}^{N-1}n_i$, where $n_i$ follows a distribution $q(w)$, we approximate the expectation with its sample mean as follows:
%\vspace{-0.12cm}
\begin{equation*}
\begin{aligned}
    & ~~ \mathbb{E}_{n_{0:N-1}}\bigg[V\bigg(\bar{\mathbf{x}}_{N|k} + \mathbf{F} \sum_{i=0}^{N-1}n_i\bigg)\bigg]  \\
    & \simeq \frac{1}{M}\sum_{m=1}^{M} V(\bar{\mathbf{x}}_{N|k} + \mathbf{F}sn_m).
\end{aligned}
%\vspace{-0.12cm}
\end{equation*}
\end{document}